%% file: hard_state_oracle.tex
\documentclass[11pt,reqno]{amsart}
\usepackage[utf8]{inputenc}
\usepackage{braket,amsmath,amsfonts,amsthm,url,graphicx,url,xcolor,hyperref}

\newcommand{\id}{\hat{1}}

\newtheorem{lemma}{Lemma}[section]
\newtheorem{prop}[lemma]{Proposition}
\newtheorem{theorem}[lemma]{Theorem}

\newtheorem{cor}[lemma]{Corollary}

\newtheorem{prob}[lemma]{Problem}
\newtheorem{rem}[lemma]{Remark}
\newtheorem{definition}[lemma]{Definition}

\newcommand{\Hil}{\mathcal{H}}
\newcommand{\ketbra}[1]{|{#1}\rangle\langle{#1}|}

\DeclareUnicodeCharacter{2009}{\,} 

\title{Quantum Oracle Separations from Complex but Easily Specified States}

\author[N. LaRacuente]{Nicholas LaRacuente}

\thanks{\hspace{-5mm} NL is supported by IBM as a Postdoctoral Scholar at the University of Chicago \& the Chicago Quantum Exchange. Contact: nlaracuente@uchicago.edu}

\begin{document}

\begin{abstract}
A foundational question in quantum computational complexity asks how much more useful a quantum state can be in a given task than a comparable, classical string. Aaronson and Kuperberg showed such a separation in the presence of a quantum oracle, a black box unitary callable during quantum computation. Their quantum oracle responds to a random, marked, quantum state, which is intractable to specify classically. We constrain the marked state in ways that make it easy to specify classically while retaining separations in task complexity. Our method replaces query by state complexity. Furthermore, assuming a widely believed separation between the difficulty of creating a random, complex state and creating a specified state, we propose an experimental demonstration of quantum witness advantage on near-term, distributed quantum computers. Finally, using the fact that a standard, classically defined oracle may enable a quantum algorithm to prepare an otherwise hard state in polynomial steps, we observe quantum-classical oracle separation in heavy output sampling.
\end{abstract}

\maketitle

\newpage
\tableofcontents
\newpage

\input{intro_2}

\section{Complex State Generation} \label{sec:hardstate}
\subsection{Background on Information and Complexity of Quantum States} \label{sec:statecomplexity}
A typical quantum state has Komogorov complexity at least polynomial in the dimension of the Hilbert space. As in \cite{mora_algorithmic_2005}, the total number of $\delta$-distinguishable (in operator norm) states is estimated at $2^{2^{n} \log \delta}$, such that for any classical description length $l$, the fraction of $l$-describable states is approximately $2^{{2^n} \log \delta + l}$. While it is not possible to verify the Kolmogorov complexity of a quantum state, it is possible to verify minimum circuit size with a fixed gate set $G$ via exhaustive search.

More definitively, to $\delta$-approximate in operator norm an arbitrary unitary in $SU(2^n)$ via $l$ gates from universal set $G$ requires
\[ l \geq \frac{(2^{2n} - 1) \log (1 / \delta) - \log c_n}{\log (2|G|)}, \]
where $c_n$ is a dimension-dependent constant such that $V(r) \leq c_n r^{{2^n}-1}$ for all $r$, and $V(r)$ is the Haar-measured volume of a ball of radius $r$ in $SU(2^n)$ \cite[eq 23]{harrow_efficient_2002}.

A simple bound in \cite{nielsen_quantum_2010} states that to $\delta$-approximate an arbitrary pure state in Hilbert space norm,
\begin{equation} \label{eq:statedensity}
l = \Omega \Big ( \frac{2^n \log (1/\delta)}{\log n} \Big ).
\end{equation}
The essential realization is that the space of pure states on $n$ qubits is the unit $(2^{n+1}-1)$-sphere, on which volumes are simpler than in $U(2^n)$. There are at most $O(2^{|G| l m})$ $\delta$-distanced states preparable from $\ket{0}^{\otimes n}$ via $l$ gates when each gate acts on $m$ input qubits. The number of states in any $\Omega(poly(n))$ circuit $\delta$-approximation scheme is $2^{O(poly(n))}$.

The Solovay-Kitaev theorem \cite{dawson_solovay-kitaev_2005} gives an upper bound on the number of universal gates required, showing that equation \eqref{eq:statedensity} is optimal up to a polynomial factor.

Equation \eqref{eq:statedensity} appears to show that even when $\delta \approx 1$, $l$ grows superpolynomially in $n$. The derivation relies however on approximating the surface area of states near a given $\ket{\psi}$ by the volume of a hypersphere in one dimension lower. For small $\delta$, the neighborhood of a state is approximately flat. As $\delta \rightarrow 1$, the approximation must be checked carefully. In principle, this might be accomplished via the distribution of distances between points on a hypersphere \cite{lord_distribution_1954, lehnen_sphere_2006} or the areas of hyperspherical caps \cite{li_concise_2011}. In practice, it is easier to work directly from results on the distribution of fidelities of quantum states \cite{zyczkowski_average_2005}. Aaronson \& Kuperberg show a similar result as Lemma 3.3 in \cite{aaronson_quantum_2006}.

In \cite{hu_characterization_2020}, it is shown that to determine if a given state has polynomial creation complexity requires exponential time in the absence of advice or witnesses (though this may include some additional assumptions). Equation \eqref{eq:statedensity} implies that a random or typical state is probably hard to prepare.

\subsection{Enumerating Complex States} \label{sec:enum}

Let $STATE\_DIAG(i, n, f, \epsilon, G)$ be a classical $\rightarrow$ quantum program that takes $i, f \in \mathbb{N}$, $\epsilon \in [0,1]$ encoded in binary, and a universal, finite gate set $G$ as its input. $STATE\_DIAG(i, n, f, \epsilon, G)$ does the following:
\begin{enumerate}
	\item Let $S$ be a set initially containing all circuits using at most $k$ gates, where $c$ is the maximum number of qubits on which any gate acts and $\ket{0}^{\otimes c k}$ the initial state.
	
	Note that set of $n$-qubit marginals of these states is included in the set of marginals on the 1st $n$ qubits by permutation invariance of the starting state and set of $k$-gate circuits.
	\item Iterate through circuits involving more than $k$ gates from $G$ to prepare $n$-qubit states, in lexograhic order. Let $C_j$ denote the $j$th such circuit. Let $l$ be a counter starting at 0.
		\begin{enumerate}
			\item If the output state of $C_j$ on $\ket{0}^{\otimes n}$ has fidelity at least $\epsilon$ with the marginal on the 1st $n$ qubits of any circuit in $S$, continue to $j+1$.
			\item If the output state of $C_j$ on $\ket{0}^{\otimes n}$ has overlap below $\epsilon$ with all 1st $n$-qubit marginals of circuits already in $S$, then (i) if $l < i$, increment $l$ and add $C_j$ to $S$ (ii) if $l = i$, return the output of $C_j$ on $\ket{0}^{\otimes n}$.
		\end{enumerate}
\end{enumerate}
Let $\ket{STATE\_DIAG(i,n,f,\epsilon, G)}$ denote the output of $STATE\_DIAG(i,n,f,\epsilon, G)$ as a quantum state. To formalize $STATE\_DIAG$ as a quantum channel, we may assume that it takes an $n$-qubit input and pinches it to the computational basis before proceeding.

The main subtlety of using $STATE\_DIAG$ is to show that it has valid states to return within the range of interest.

\subsection{Existence of Marked States}
\begin{lemma} \label{manystates}
In a system of $n$ qubits with $\epsilon \in (0,1]$ and for a set of densities $S = \{\rho_j\}_{j=1}^m$ with total rank $r : m \leq r \leq m 2^n$, there is a set of at least
\begin{equation} \label{eq:avstates}
\lfloor (1-\epsilon)^{1-2^n} - r \rfloor \geq \lfloor \exp(\epsilon(2^n-1)) - r \rfloor
\end{equation}
pure states such that for any $\ket{\psi}, \ket{\phi}$ in the set, $|\braket{\psi | \phi}|^2 < \epsilon$, and for any $\ket{\psi}$ in the set and $\rho \in S$, $F(\ketbra{\psi}, \rho) < \epsilon$. The probability that a random pure state $\ket{\psi}$ has $F(\ketbra{\psi}, \rho) < \epsilon$ for all $\rho \in S$ is at least
\begin{equation} \label{eq:stateprob}
1 - r (1-\epsilon)^{2^n-1} \geq 1 - r \exp(- \epsilon (2^n-1)) .
\end{equation}
\end{lemma}
\noindent Lemma \ref{manystates} replaces equation \eqref{eq:statedensity} for $\epsilon$ arbitrarily close to 1. It is similar to Aaronson \& Kuperberg's Lemma 3.3, though precedent for this result appears in \cite{kus_universality_1988, zyczkowski_average_2005}.
\begin{proof}
The distribution of fidelity between pure states in the Fubini-Study/$U(2^n)$-invariant Haar measure is given by
\begin{equation}
p_N(F) = (N-1) (1 - F)^{N-2},
\end{equation}
where here $N = 2^n$ \cite{kus_universality_1988, zyczkowski_average_2005, aaronson_quantum_2006}. Integrating the measure,
\begin{equation} \label{eq:fidb}
p_N(F < \epsilon) = 1 - (1-\epsilon)^{N-1} \text{, and } p_N(F > \epsilon) = (1-\epsilon)^{N-1}.
\end{equation}
While this bound is originally for a pair of random densities, as noted in \cite{zyczkowski_average_2005}, the symmetries of the state space allow us to fix one of the densities.

For an any $\ket{\psi}$ and diagonalized density $\rho = \sum_i \alpha_i \ketbra{\phi_i}$,
\begin{equation}
F(\ketbra{\psi}, \rho) = \braket{\psi | \rho | \psi} = \sum_i \alpha_i |\braket{\psi | \phi_i}|^2 \leq \max_i |\braket{\psi | \phi_i}|^2 .
\end{equation}
Hence we may consider the set of at most $r$ pure states $\{\ketbra{\phi_j}\}_{j=1}^{r}$, where each $\ket{\phi_j}$ is a diagonal basis state of some density in $S$. By the union bound and equation \eqref{eq:fidb},
\begin{equation}
p_N(\exists j \in 1...m : F(\ketbra{\psi}, \rho_j) > \epsilon) \leq p_N(\exists j \in 1...r : |\braket{\psi | \phi_j}|^2 > \epsilon) \leq r  (1-\epsilon)^{N-1}
\end{equation}
for $\ket{\psi}$ chosen uniformly at random from the $N$-dimensional Hilbert space.

If we have already chosen some $\{\ket{\psi}_k\}_{k=1}^K$, then we may add these to the states to avoid. While the probability of finding a state with at most $\epsilon$ overlap decreases, this does not prevent a search from continuing to find states as long as there is non-zero probability of a state being sufficiently far from the avoided set. Hence the maximum number $K$ of pure states with mutual overlap no more than $\epsilon$ from each other or from any of the original $m$ densities is lower bounded as
\begin{equation*}
K \geq \lfloor  (1-\epsilon)^{1-N} - r \rfloor .
\end{equation*}
For the exponential approximation, we recall that $= (1-\epsilon)^{N-1} = e^{(N-1) \ln (1 - \epsilon)} \leq e^{-\epsilon (N-1)}$.
\end{proof}
\begin{rem}
The fidelity of any pure state with the complete mixture is $1/2^n$, and the average overlap between pure states in the Fubini-Study measure is also shown to be $1/2^n$ in \cite{kus_universality_1988, zyczkowski_average_2005}. This point corresponds with $\epsilon = O(2^n)$, at which the state number bound of Lemma \ref{manystates} may become negative. When $\epsilon \leq 2^{- \alpha n}$ for $\alpha < 1$, the number of available states grows double-exponentially with qubit number. Equation \eqref{eq:statedensity} scales similarly with $\alpha$ and $n$.
\end{rem}
\begin{rem}
We may rewrite the available state number bound given by equation \eqref{eq:avstates} as
\[ \exp( 2 (2^n - 1) \ln (1/\delta) ) - r\]
and the probability bound of \eqref{eq:stateprob} as $1 - \exp(2 (2^n - 1) \ln (\delta) (1 + \ln(r)))$, where $\delta = \sqrt{1 - \epsilon}$ bounds trace distance. The state number bound is then comparable to equation \eqref{eq:statedensity}, up to a dimension-dependent constant due to differences in the norm considered.
\end{rem}

\begin{lemma} \label{statehalt}
Let $g(n), f(n), s(n) : \mathbb{N} \rightarrow \mathbb{N}$ be functions given such that $\omega(poly(n)) < g(n) \leq o(2^n)$, and $\omega(poly(n)) < f(n) < o(\log g(n))$. Then there are at least
\begin{equation}
\lfloor \exp(g(n)^{-1} (2^n - 1)) - \exp(c f(n) \ln(c f(n))) \rfloor = \Omega(\exp(g(n)^{-1} 2^n))
\end{equation}
values of $i$ such that $STATE\_DIAG(n,f(n),g(n)^{-1}, i)$ definitely halts.
\end{lemma}
\begin{proof}
Via Lemma \ref{manystates}, there are at least
\begin{equation}
\lfloor \exp(g(n)^{-1} (2^n - 1)) - r \rfloor
\end{equation}
states with distance at least $g(n)^{-1}$ from those preparable, where we must now calculate $r$ from $f(n)$.

Let $G$ be a fixed gate set in which each gate acts on a maximum of $c$ qubits. A gate sequence of length $f(n)$ may act only on $c f(n)$ qubits in total. Hence the number of preparable states is no more than than the total number of sequences of gates, each choosing at most $c$ inputs. We calculate
\begin{equation}
r \leq (c f(n))^{c f(n)} = \exp(c f(n) \ln(c f(n))) < \exp(b \times o(\log g(n))) \leq o(poly(g(n)))
\end{equation}
for some constant $b$. We see that $r$ is not of the same order as the term from which it subtracts.
\end{proof}

\section{Proof of Hardness of MQST} \label{sec:mqst}

\begin{rem} \normalfont \label{oneoracle}
To understand the idea of the proof, one may consider the case of a single query using one pure state $\ket{\phi}$ on only $n$ qubits. We decompose
\begin{equation} \label{decomp1}
\ket{\phi} = \braket{\psi | \phi} \ket{\psi} + \sum_{j=2}^{N} \braket{\gamma_j | \phi} \ket{\gamma_j} ,
\end{equation}
where $N = 2^n$, and $\{ \ket{\gamma_j} \}$ form a basis with $\ket{\gamma_1} = \ket{\psi}$. Let $V$ correspond to the 1st possible case of $U$, which applies a phase of $-1$ to $\ket{\psi}$. The inner product of $V \ket{\phi}$ with $\ket{\phi}$ is given by
\begin{equation}
\braket{\phi | V | \phi} =  \sum_{j=2}^N |\braket{\gamma_j | \phi}|^2 - |\braket{\psi | \phi}|^2 = 1 - 2 |\braket{\psi | \phi}|^2
\end{equation}
for normalized amplitudes. Hence
\begin{equation}
|\braket{\phi | V | \phi}|^2 = 1 - 4 |\braket{\psi | \phi}|^2 + 4 |\braket{\psi | \phi}|^4 .
\end{equation}
Via the well-known  Fuchs–van de Graaf (in)equalities,
\begin{equation} \label{fvdg}
\frac{1}{2} \| V \ketbra{\phi} V^\dagger - \ketbra{\phi} \|_1 = 2 |\braket{\psi | \phi}| \sqrt{1 - |\braket{\psi | \phi}|^2} .
\end{equation}
The operational interpretation of the trace distance then bounds the probability of successfully distinguishing between the states $V \ket{\phi}$ and $\ket{\phi}$, where the latter is equivalent to applying the identity.
\end{rem}

In general, Arthur is neither limited to preparing and testing a single state against the oracle, nor to using the quantum oracle only once on a given state. The main subtlety is that while $STATE\_DIAG$ checked that the marked state is hard to prepare using gate set $G$, Arthur may use the oracle $U$ as though it were an extra gate. Formally, we require the marked state to be hard to prepare using gate set $G \cup \{U\}$, which we have not yet checked. Marked states might be chosen from a small set of easily codeable indices, so it's not obvious that the original Aaronson and Kuperberg counting arguments \cite{aaronson_quantum_2006} apply.

Intuition from constraints on search \cite{zalka_grovers_1999, boyer_tight_1998} suggests that Arthur shouldn't be able to amplify overlap of a prepared state with the marked state too quickly, even including $U$ in computation. The first step to make this intuition rigorous is to show that even if $U$ is not the identity, its effect on any state with small overlap with the marked state is small. Then we may apply a standard argument from the study of query complexity, showing that a polynomial length, alternating sequence of calls to $U$ and circuits of gates from $G$ neither prepares a state that is close to the marked state nor distinguishes case (1) of MQST from case (2) with high probability. Departing from the usual argument, we actually show a weaker separation than query complexity when marked states are easily codeable, as exponential-length circuits with few queries might be sufficient. Unlike full query complexity separation, the circuit length separation holds in a wider range of regimes and may appear in simple physical situations as described in Section \ref{sec:easy}.
\begin{lemma} \label{tenover}
Let $N = 2^n$, $\ket{\psi_j}_{j=1}^J \in \Hil_{N}$ form a basis of a given subspace $S \subseteq \Hil_{N}$, and $\ket{\gamma_j}_{j=J}^{N}$ form a basis of $S^\perp$. Let $V$ be defined by $V \ket{\psi_j} = - \ket{\psi_j}$ and $V \ket{\gamma_j} = \ket{\gamma_j}$. Then for any pure state $\ket{\phi}$ on $r > n$ qubits on which $V$ is extended to $V \otimes \id_{n-r}$,
\begin{equation}
\braket{\phi | V | \phi} = 1 - 2 \sum_{j=1}^J \sum_{i=1}^{N} \lambda_i |\braket{\eta_i | \psi_j}|^2 = 1 - 2 \sum_{j=1}^J F(\rho, \ketbra{\psi_j}) ,
\end{equation}
where $\rho$ is the marginal of $\ket{\phi}$ on the 1st $n$ qubits, $\{\ket{\eta_i}\}$ are the eigenstates of $\rho$, and $\{\lambda_i\}$ are the associated eigenvalues.

\end{lemma}
\begin{proof}
We Schmidt decompose $\ket{\phi}$ and rewrite it in terms of a basis $\{\ket{\gamma_j}\}_{j=1}^{N}$ such that $\ket{\gamma_1} = \ket{\psi}$.
\begin{equation} \label{eq:schmidt1}
\ket{\phi} = \sum_i \alpha_i  \ket{\eta_i} \otimes \ket{\sigma_i} = \sum_i \alpha_i 
	\Big ( \sum_{j=1}^J \braket{\psi_j | \eta_i} \ket{\psi} + \sum_{j=J}^N \braket{\gamma_j | \eta_i} \ket{\gamma_j} \Big ) \otimes \ket{\sigma_i} .
\end{equation}
Via direct calculation,
\begin{equation}
\begin{split}
\braket{\phi | V | \phi} & = \sum_{i,k} \alpha_i \alpha_k^*
	\Big ( - \sum_{j,l = 1}^J \braket{\psi_j | \eta_i} \braket{\eta_k | \psi_l} \braket{\psi_j | \psi_l} +
	\sum_{j,l = J}^N \braket{\gamma_j | \eta_i} \braket{\eta_k | \gamma_l} \braket{\gamma_l | \gamma_j} \Big )
	\braket{\sigma_i | \sigma_k} \\
	& = \sum_i |\alpha_i|^2 \Big (\sum_{j=J}^N |\braket{\eta_i | \gamma_j}|^2 - \sum_{j=1}^J |\braket{\eta_i | \psi_j}|^2 \Big)
\end{split}
\end{equation}
The second inequality follows primarily from state orthogonality. We note that
\begin{equation}
\braket{\eta_i | \eta_i} = 1 = \sum_{j=J}^N |\braket{\eta_i | \gamma_j}|^2 + \sum_{j=1}^J |\braket{\eta_i | \psi_j}|^2 ,
\end{equation}
so for each $i$,
\begin{equation}
\sum_{j=J}^N |\braket{\eta_i | \gamma_j}|^2 - \sum_{j=1}^J |\braket{\eta_i | \psi_j}|^2
	= 1 - 2 \sum_{j=1}^J |\braket{\eta_i | \psi_j}|^2 ,
\end{equation}
yielding the first equality of the Lemma, where $\lambda_i = |\alpha_i|^2$. We also have that
\begin{equation}
F(\rho, \ketbra{\psi_j}) = \braket{\psi_j | \rho | \psi_j} = \sum_i |\alpha_i|^2 |\braket{\eta_i | \psi_j}|^2 ,
\end{equation}
which completes the Lemma.

\end{proof}

\begin{cor} \label{tenoverext}
Let $V$, $r$, $n$, and $\ket{\psi}$ be as in Lemma \ref{tenover}. For given $S_1 \subset \Hil_{2^n}$ such that $\ket{\psi} \perp span\{S_1\}$, let $\tilde{V}$ be defined by $\tilde{V} \ket{\eta} = - \ket{\eta}$ for all $\ket{\eta} \in span\{S_1\}$, and $\tilde{V} \ket{\eta} = \ket{\eta}$ for all $\ket{\eta} \perp span\{S_1\}$. Then
\begin{equation}
\braket{\phi | \tilde{V}^\dagger V \tilde{V} | \phi} = 1 - 2 F(\rho, \ketbra{\psi}) = \braket{\phi | V | \phi},
\end{equation}
where $\rho$ is the marginal of $\phi$ on the 1st $n$ qubits.
\end{cor}
\begin{proof}
$\tilde{V}$ commutes with a partial trace over the last $r-n$ qubits. Applying $\tilde{V}$ to $\ket{\phi}$ as decomposed in equation \eqref{eq:schmidt1}, we see that $\tilde{V}$ acts as the identity on $\ket{\psi}$ and preserves orthogonality of states, so we may absorb its effect into the definition of $\ket{\gamma_j}$. The sum $\sum_i |\alpha_i|^2 |\braket{\eta_i | \psi}|^2$ is not affected by $\tilde{V}$, so by Lemma \ref{tenover},
\[ F(\rho, \ketbra{\psi}) = F(\tilde{V} \rho \tilde{V}^\dagger, \ketbra{\psi}) . \]
The Corollary then follows the Lemma for $\tilde{V} \ket{\phi}$.
\end{proof}

\begin{lemma} \label{indist}
For given $k, n \in \mathbb{N}$, let $N = 2^n$, $\ket{\psi_{j}}_{j=1}^k \in \Hil_{N}$ be a pure state such that $S \subseteq \Hil_{N}$, and $\ket{\gamma_{m,j}}_{m=2}^{N}$ form a basis of $S^\perp$ for each $j$. For each $j \in 1...k$, let $V_j$ be defined by $V_j \ket{\psi_j} = - \ket{\psi_j}$, and $V_j \ket{\gamma_{m,j}} = \ket{\gamma_{m,j}}$ for $m > 1$. Let $r \geq n$, and $\ket{\textbf{0}} = \ket{0...0} = \ket{0}^{\otimes r}$. Let $T_1, ..., T_k$ be isometries such that for any $l \in 1...k$, $T_l ... T_1 \ket{\textbf{0}} \in \Hil_r$ is within the set of states such that
\[ \sum_j F(\rho_l, \ketbra{\psi_j}) \leq \epsilon < 1/2 , \]
where $\rho_l$ is the marginal of $T_l ... T_1 \ket{\textbf{0}}$ on the 1st $n$ qubits. Then
\begin{equation}
\frac{1}{2} \|T_k ... T_1 \ketbra{\mathbf{0}} T_1^\dagger ... T_k^\dagger - V_k T_k ... V_1 T_1 \ketbra{\mathbf{0}} T_1^\dagger V_1^\dagger ... T_k^\dagger V_k^\dagger \|_1
	\leq \sqrt{2\sqrt{2 k \epsilon} + 4 k \epsilon - 2 (2 k \epsilon)^{3/2}}.
\end{equation}
which bounds the probability of distinguishing for a given unitary family $U_1, ..., U_k$ whether $U_1, ..., U_k =\id$ or $U_j = V_j$ for each $j \in 1...k$ using circuits of the form $U_k T_k ... U_1 T_1 \ketbra{\mathbf{0}} T_1^\dagger U_1^\dagger ... T_k^\dagger U_k^\dagger$ .
\end{lemma}
\begin{proof}

We will show this Lemma by induction, using the fact that after applying $V_j$ to a state that has low overlap with $\ket{\psi_j}$, the new state has high overlap with the original. First, for any $l \in 1...k$,
\begin{equation}
V_l T_l ... T_1 \ket{\mathbf{0}} = \braket{\mathbf{0} |T_1^\dagger ... T_l^\dagger V T_l ... T_1 | \mathbf{0}} T_l ... T_1 \ket{ \mathbf{0}} 
	+ \beta_l V_l \ket{\tilde{\phi}_l}
\end{equation}
for some normalizing $\beta_l$ and $\ket{\tilde{\phi}_l}$ that is orthogonal to $T_l ... T_1 \ket{ \mathbf{0}}$. Hence
\begin{equation}
\begin{split}
V_{l+1} T_{l+1} V_l T_l ... T_1 \ket{\mathbf{0}} =.
	 & \braket{\mathbf{0} |T_1^\dagger ... T_l^\dagger V_l T_l ... T_1 | \mathbf{0}}
	\big ( \braket{\mathbf{0} |T_1^\dagger ... T_{l+1}^\dagger V_{l+1} T_{l+1} ... T_{1}
		| \mathbf{0}} T_{l+1} ... T_1 \ket{\mathbf{0}}  
	\\ & + \beta_{l+1} T_{l+1} V_{l+1} \ket{\tilde{\phi}_{l+1}} \big ) + \beta_l V_{l+1} T_{l+1} V_l \ket{\tilde{\phi}_l} ,
\end{split}
\end{equation}
which we rewrite as
\begin{equation*}
\begin{split}
V_{l+1} T_{l+1} V_l T_l ... T_1 \ket{\mathbf{0}} = \braket{\mathbf{0} |T_1^\dagger ... T_l^\dagger V_l T_l ... T_1 | \mathbf{0}}
	\braket{\mathbf{0} |T_1^\dagger ... T_{l+1}^\dagger V_{l+1} T_{l+1} ... T_{1} | \mathbf{0}} T_{l+1} ... T_1 \ket{\mathbf{0}}  
	+ \xi_{l+1} \ket{\eta_{l+1}}
\end{split}
\end{equation*}
for some $\xi_{l+1}$ and $\ket{\eta_{l+1}}$. Let $\zeta_l = \braket{\mathbf{0} |T_1^\dagger ... T_l^\dagger V_l T_l ... T_1 | \mathbf{0}}$. By induction starting from $T_1 \ket{\mathbf{0}}$,
\begin{equation}
V_k T_k V_{k-1} T_{k-1} ... V_1 T_1 \ket{\mathbf{0}} = \Big ( \prod_{l=1}^k \zeta_l \Big ) T_k ... T_1 \ket{\mathbf{0}} + \beta \ket{\eta}
\end{equation}
for some normalizing $\beta$ and state $\ket{\eta}$, which need not be orthogonal to $T_k ... T_1 \ket{\mathbf{0}}$. By Lemma \ref{tenover},
$\zeta_l \geq 1 - 2 \epsilon$ for any $l \in 1...k$. Hence
\begin{equation}
|\braket{\mathbf{0} |T_1^\dagger ... T_{k}^\dagger V_k T_k ... V_1 T_1 | \mathbf{0}}| \geq (1 - 2 \epsilon)^k - \sqrt{1 - (1 - 2 \epsilon)^k} .
\end{equation}
Via Bernoulli's inequality and assuming $2 \epsilon < 1$,
\begin{equation}
|\braket{\mathbf{0} |T_1^\dagger ... T_{k}^\dagger V_k T_k ... V_1 T_1 |
	\mathbf{0}}| \geq 1 - 2 k \epsilon - \sqrt{2 k \epsilon} .
\end{equation}
Therefore,
\begin{equation}
F(T_k ... T_1 \ketbra{\mathbf{0}} T_1^\dagger ... T_k^\dagger \textbf{,} V_k T_k ... V_1 T_1 \ketbra{\mathbf{0}} T_1^\dagger V_1^\dagger ... T_k^\dagger V_k^\dagger) \geq 1 - 2\sqrt{2 k \epsilon} - 4 k \epsilon + 2 (2 k \epsilon)^{3/2}.
\end{equation}
Since both states are pure, the fidelity is exactly related to the trace distance, yielding
\begin{equation}
\frac{1}{2} \|T_k ... T_1 \ketbra{\mathbf{0}} T_1^\dagger ... T_k^\dagger - V_k T_k ... V_1 T_1 \ketbra{\mathbf{0}} T_1^\dagger V_k^\dagger ... T_k^\dagger V_1^\dagger \|_1
	\leq \sqrt{2\sqrt{2 k \epsilon} + 4 k \epsilon - 2 (2 k \epsilon)^{3/2}}.
\end{equation}
This trace distance is between the $U_k, ... U_1 = V_k, ..., V_1$ and $U_1, ..., U_k = \id$ case. That it bounds distinguishability follows from the operational interpretation of the trace norm.
\end{proof}
\begin{rem}
We could easily extend Lemma \ref{indist} to marked subspaces from marked states, but in this work there is no reason to do so, and it would complicate the analysis.
\end{rem}

\subsection{State Complexity in Computation}
In this section, we show the technical definition of MQST and proof of its hardness for a polynomial quantum computation with possibly unbounded classical pre-computation and witnesses.

Polynomial classical information at the start of computation, including witnesses, inputs, and advice, can be expressed as an extra $O(poly(n))$ circuit that prepares the classical state as a quantum state in the computational basis. With polynomial overhead, we may thereby include all possibilities for polynomial classical witnesses, inputs, and advice in the set of circuits considered by $STATE\_DIAG$.

When all resources are polynomially bounded, including classical witnesses and pre-computation, it is a standard procedure to absorb classical resources into the quantum computation. We would like, however, to allow classical witnesses and pre-computation to be superpolynomially large while restricting their quantum counterparts to polynomial size or time. To do so, we will formulate the problem in a way that separates these classes of resource.

Even when classical witnesses and pre-processing are superpolynomial, a polynomial quantum circuit included with this computation accesses only polynomially-many classical bits. Hence we may re-express these bits with an extra $O(poly(n))$ circuit that depends on results of classical pre-computation. Adding this to the original, quantum circuit, we replace a polynomial quantum circuit $T$ by a conditional, still polynomial quantum circuit $T_w$, where $w$ is a classical bit string that may depend on all classical information available at the start of computation. Hence instead of trying to formulate the polynomial quantum circuit as running on a larger classical input, we equivalently think of applying a polynomial circuit with polynomial input that depends on available classical information.
\begin{theorem}[Technical Version of \ref{thm:main}] \label{mqst}
Let the sequence $(i_n)_{n=1}^\infty$ be defined by choosing $i_n$ from a sequence of non-empty sets $(\{i\}_n)_{n=1}^\infty$ on which $STATE\_DIAG(n,i)$ halts. For all $n$, there exist valid $\{i\}_n$ of cardinality $|\{i\}_n| = \Omega(2^{2^n})$. Let $(U_n)_{n=1}^\infty$ be a family of marked state oracles, where the $n$th marked state is equal to $STATE_DIAG(n,i)$ for $i \in \{i\}_n$. By $|\{i\}_n|$ we denote the cardinality of set $\{i\}_n$. For sufficiently large $n$ and with unbounded quantum circuit size,
\begin{enumerate}
	\item if $|\{i\}_n| = \Omega(2^{2^n})$, then $MQST$ generally requires an exponential classical or polynomial quantum proof of case (1) to verify with fewer than $\Omega(\sqrt{2^n} / poly(n))$ queries. It is nonetheless verifiable with one query given either such proof;
	\item If $\max_i \{i\}_n = O(2^n)$, then one query may verify a (polynomially long, classical) binary representation of $i_n$ as a proof of case (1);
	\item If $\max_i \{i\}_n \leq O(poly(n))$, then a trusted, polynomial, classical description of $\{i\}_n$ $MQST(n, i_n)$ is sufficient to perform MQST with $O(poly(n))$ queries. It is similarly possible to verify an untrusted such description as a proof of case (1).
\end{enumerate}
All of the above cases require superpolynomial quantum computation.
\end{theorem}
\begin{proof}[(Proof of Theorem \ref{mqst})]
By Lemma \ref{statehalt}, valid marked states $\ket{\psi}$ exist for all sufficiently large $n$.

To fully formalize the task, assume the Arthur starts with the state $\ket{0...0}$ on $r = O(poly(n))$ qubits, plus a classical auxiliary systems in a configuration that may not depend on whether $U$ is the identity, but may convey any other information, including what the marked stat would be. Arthur may apply gates from the set $G$, single-qubit measurements in the computational basis, the oracle to any $n$-qubit subsystem, and partial traces to subsets of qubits.

Let $X$ denote the set of states of Arthur's classical system, in state $x \in X$ after receiving a classical proof from Merlin and an arbitrary amount of classical pre-computation, but before any quantum operations. It is possible that $x$ classically describes what the marked state would be if $U$ is not the identity.

All steps in computation after classical precomputation but before final traces have (up to equivalence) the form
\[ T_{k+1, x} U T_{k, x} U ... U T_{1, x} , \]
where $T_{1,x} ... T_{k+1,x}$ are unitaries that may depend on $x$, and $k = O(poly(n))$. By the definition of $STATE\_DIAG$ and the parameters of the problem, the sequence $T_{l,x} ... T_{1,x}$ results in a state that has at most fidelity $\epsilon$ with $\ket{\psi}$. We apply Lemma \ref{indist} to $U T_{k,x} U ... U T_{1,x}$, and by approximating away higher orders in $\epsilon$, obtain
\begin{equation} \label{eq:nb}
\frac{1}{2} \|T_{k,x} ... T_{1,x} \ketbra{\mathbf{0}} T_{1,x}^\dagger ... T_{k,x}^\dagger - V T_{k,x} ... V T_{1,x} \ketbra{\mathbf{0}} T_{1,x}^\dagger V^\dagger ... T_{k,x}^\dagger V^\dagger \|_1	\leq 3 \sqrt[4]{k \epsilon}
\end{equation}
where $V$ is the non-identity case of $U$. The operational interpretation of the 1-norm bound then shows that with access to this state, Arthur can't correctly guess whether $U = V$ or $U = \id$ with probability higher than $\epsilon$. Convexity of the 1-norm implies that a mixture of $x$ achieves no higher value, and the related monotonicity under channels shows that neither does $T_{k+1}$ nor the differed traces.

For the particular cases, the cardiality $|\{i\}_n|$ determines the number of bits needed either to enumerate all possible marked states or to specify a particular marked state classically.
\end{proof}

\begin{rem}
What if Arthur is allowed arbitrary classical computation interspersed with elements of a polynomially-bounded quantum circuit? This case is left open. The proof strategy of Theorem \ref{mqst} no longer applies, as it is not possible to dilate arbitrary, intermediate classical processing within polynomially bounded quantum space and time. Arthur could measure intermediate states and apply superpolynomial classical processing to determine subsequent unitaries, and it is not entirely clear whether these measurements could be deferred to the end without a superpolynomial expansion in the quantum circuit size.
\end{rem}

\section{Computational Complexity} \label{sec:complexity}

\begin{definition}[Marked Quantum State Oracle (MQSO)] \label{def:mqsp}
Let $L$ be a given, binary language. Let the sequence $(i_n)_{n=1}^\infty$ be defined by choosing $i_n$ from a sequence of non-empty sets $(\{i\}_n)_{n=1}^\infty$ on which $STATE\_DIAG(n,i)$ halts, and let $\ket{\psi_n} = \ket{STATE\_DIAG(n, i_n)}$, and $U_n : \Hil_{2^n} \rightarrow \Hil_{2^n}$ be defined by
\[ U_n \ket{\psi_n} \otimes \ket{x}=
	(-1)^{L(x)} \ket{\psi_n} \otimes \ket{x}, \]
and $U_n \ket{\phi} \otimes \ket{x} = \ket{\phi} \otimes \ket{x}$ for all $n$-qubit $\ket{\phi} \perp \ket{\psi_n}$. Let $MQSO(L, \{i_n\})$ denote the constructed oracle family $\{U_n\}$. We restrict to parameter ranges for which the indexed marked state exists.
\end{definition}
\begin{cor}[Technical Version of Cor \ref{cor:complexity}] \label{mqsl}
Let the sequence $(i_n)_{n=1}^\infty$ be defined by choosing $i_n$ from a sequence of non-empty sets $(\{i\}_n)_{n=1}^\infty$ on which $STATE\_DIAG(n,i)$ halts. For all $n$, there exist valid $\{i\}_n$ of cardinality $|\{i\}_n| = \Omega(2^{2^n})$. Let $L$ be a language constructed by randomly choosing $\lfloor 2^n/2 \rfloor$ of $2^n$ possible $n$-length strings for each $n$, and let $\{i\}_n$ be a sequence of index sets such that $i_n \in \{i\}_n$ for all $n$. With access to the family of marked state oracles $MQSO(L, \{i_n\})$. Then
\begin{enumerate}
	\item $L \not \in QCMA^{(U_n)}/poly$. This holds even with the promise that $|\{i\}_n |= O(1)$ for all $n$. Even if we relax the complexity class to allow unbounded, classical, untrusted witness size and unbounded, classical pre-computation, polynomial quantum circuits (including oracle queries) and polynomial, classical advice fail to decide $L$ with bounded error. 
	\item $L \in QMA^{(U_n)}$.
	\item With the promise that $|\{i\}_n| = O(poly(n))$, $L \in BQP^{(U_n)}/qpoly$.
\end{enumerate}
Hence $\exists$ a quantum oracle family $(U_n)_{n=1}^\infty$ such that $BQP^{(U_n)}/qpoly \not \subset QCMA^{(U_n)}/poly$.
\end{cor}
\noindent A \& K showed as their primary result that $QMA^{(U_n)} \not \subset QCMA^{(U_n)}$ for a related quantum oracle as theorem 1.1. Though A \& K did not claim $BQP^{(U_n)}/qpoly \not \subset QCMA^{(U_n)}/poly$, this is not surprising given their arguments for $BQP^{(U_n)}/qpoly \not \subset BQP^{(U_n)}/poly$ in theorem 3.4. Again more surprising in our formulation is the lack of reliance on query complexity.
\begin{rem} \label{noquery}
With unbounded quantum circuit size, $L$ as defined in Theorem \ref{mqsl} is decidable with $MQSO(L, \{i_n\})$ and...
\begin{enumerate}
	\item An exponential, classical or polynomial, quantum witness and one query, or exponentially many queries, if $|\{i\}_n| = \Omega(2^{2^n})$.
	\item A polynomial, untrusted witness and one query if $\max_i \{i\}_n = O(2^n)$ for all $n$.
	\item Polynomial queries and a (trusted or untrusted) description of $\{i\}_n$ if $\max_i \{i\}_n = O(poly(n))$.
	\item One query if $\max_i \{i\}_n = 1$ and $\{i\}_n$ is given by a computable function for each $n$.
\end{enumerate}
\end{rem}

\begin{lemma} \label{advicebound}
Let $L$ be a binary language, for which membership of strings $x \in \{0,1\}^n$ in $L$ is determined arbitrarily (having no pre-determined dependence on the value of $x$). For given $n$ and $z$ of length $O(poly(n))$, there are at least $\Omega(2^n/poly(n))$ values of $x$ for which $L(x)$ cannot be determined from $x \otimes z$.
\end{lemma}
\begin{proof}
This Lemma follows from a counting argument mirroring that of Aaronson \& Kuperberg. There are $2^{2^n}$ possible assignments of $L(x)$, and $2^{O(poly(n))}$ possible values of $z$. Hence at least 1 value of $z$ distinguishes a class of at least $2^{2^n / O(poly(n))}$ possible assignments. Within this class, at least $\Omega(2^n / poly(n))$ values of $x$ must not be fixed.
\end{proof}
In Lemma \ref{advicebound}, we do not know that conditioned on $z$, the $\Omega(2^n/poly(n))$ may vary independently. It might be that the class implied by $z$ forces strong correlations between whether particular strings are in the language. These correlations will not matter too much for the complexity of determining whether a particular $x \in L$.

\begin{proof}[(Proof of Corollary \ref{mqsl} part 1)]
Let $(U_n)$ be a quantum oracle family as in Definition \ref{def:mqsp}. Let $x$ denote a particular input string.

Let $W$ be the set of Merlin's possible, classical responses for a given $n$, and $w$ be Merlin's particular response in a given case, which may depend on the input $x$ and on randomness as well as on the language and oracle. Let $z \in Z$ analogously denote the respective particular quantum advice given for input size $n$, and set of possible advice states. As in the proof of Theorem \ref{mqst}, after Merlin's response is given, Arthur's computation can be modeled as a polynomial quantum circuit of the form $T_{k+1, w, x, z} U_{m_k} T_{k, w, x, z} U_{m_{k-1}} ... U_{m_1} T_{1, w, x, z}$ for $k = O(poly(n))$, followed by partial traces, where $m_1, ..., m_{k}$ index calls to oracles in the given family.

In the starting configuration with Merlin's witness, Arthur's state has the form $x \otimes w \otimes z \otimes \ket{0...0}$. The witness and input $x \otimes w$ don't distinguish between $x \in L$ and $x \not \in L$ when Merlin is untrusted (the binary string representing $x$ contains no information about whether $x \in L$, so with no advice or oracle access, $L$ is undecidable even with unbounded time). The advice state $z$ may depend on whether $x$ is in the language, but it is the same value for all possible $x$. By Lemma \ref{advicebound}, there are at least $\Omega(2^n/poly(n))$ possible values of $x$ for which $z$ fails to distinguish whether $x \in L$. We may restrict our attention to these values of $x$.

Formally, the trace distance between the $x \in L$ and $x \not \in L$ is zero, and it upper bounds the completeness-soundness gap. We may easily subsume $x,z,w$ into
\[T_{k+1, w, x, z} U_{m_k} T_{k, w, x, z} U_{m_{k-1}} ... U_{m_1} T_{1, w, x, z} \]
by allowing any of Arthur's non-oracle unitaries to prepare $x \otimes w$ in the computational basis if desired, so it is not necessary to retain these explicitly as separate parts of the state. Hence we aim to show that
\[ T_{k+1, w, x, z} U_{m_k} T_{k, w, x, z} U_{m_{k-1}} ... U_{m_1} T_{1, w, x, z} \ket{0}^{\otimes r} \]
has small difference as measured by the trace distance between the case in which $x \in L$ and that in which $x \not \in L$ for the same value of $w$, regardless of what $w$ actually is. Hence if Merlin can send a classical $w$ that convinces Arthur that $x \in L$, the same $w$ will probably trick Arthur if in truth $x \not \in L$.

For any $l \in 1...k$, let $\rho_l$ be the marginal on the 1st $2 n$ qubits of $ T_{l, w, x, z} ... T_{1, w, x, z} \ket{0}^{\otimes r}$, $\eta_{l,i}$ be the $i$th eigenvector of $\rho_l$, and $\lambda_i$ the corresponding eigenvalue. Let $\ket{\psi_m}$ be the marked state for any $m$. Then
\[ \sum_{y \in L} \sum_i F(\rho_l, \ketbra{\psi_m} \otimes \ketbra{y})
	= \sum_{y \in L} \sum_i \lambda_i |\braket{\psi_m, y | \eta_{l,i}}|^2 
	= \sum_{y \in L} \sum_i \lambda_i \braket{\eta_{l,i} | \psi_m, y}  \braket{\psi_m, y | \eta_{l,i}} . \]
Let \[\eta_{i,l} = \sum_j \alpha_j \ket{\sigma_{l,i,j}} \otimes \ket{\gamma_{l,i,j}} \]
in a Schmidt decomposition. We then have from the above sum,
\[ ... = \sum_i \sum_{y \in L} \lambda_i \sum_{j,a} \alpha_{l,i,j}^* \alpha_{l,i,a}
	\braket{\sigma_{l,i,j} | \psi_m} \braket{\psi_m | \sigma_{l,i,a}}
	\braket{\gamma_{l,i,j} | y } \braket{ y | \gamma_{l,i,a} } . \]
Comparing $\sum_{y \in L} \ketbra{y}$ to the identity for each $l,i,j,a$,
\[ \delta_{j,a} = \sum_{y \in \{0,1\}} \braket{\gamma_{l,i,j} | y } \braket{ y | \gamma_{l,i,a} }
	= \sum_{y \in L} \braket{\gamma_{l,i,j} | y } \braket{ y | \gamma_{l,i,a} } + \sum_{y \not \in L} \braket{\gamma_{l,i,j} | y } \braket{ y | \gamma_{l,i,a} }. \]
We note that
\[ \sum_i \sum_{y \not \in L} \lambda_i \sum_{j,a} \alpha_{l,i,j}^* \alpha_{l,i,a}
	\braket{\sigma_{l,i,j} | \psi_m} \braket{\psi_m | \sigma_{l,i,a}}
	\braket{\gamma_{l,i,j} | y } \braket{ y | \gamma_{l,i,a} }
	= \sum_{y \not \in L} \sum_i F(\rho_l, \ketbra{\psi_m} \otimes \ketbra{y}) \geq 0, \]	
so by replacing $\sum_{y \in L} \ketbra{y}$ by the identity,
\[ ... \leq \sum_i \sum_{j} \lambda_i |\alpha_{l,i,j}|^2
	\braket{\sigma_{l,i,j} | \psi_m} \braket{\psi_m | \sigma_{l,i,j}} = \sum_i \sum_j \lambda_i |\alpha_{l,i,j}|^2 |\braket{\sigma_{l,i,j} | \psi_m}|^2 . \]
Let $\tilde{\rho}_{l,i}$ be the 1st $n$-qubit marginal of $\eta_{l,i}$, which has eigenvectors $\{\sigma_{l,i,j}\}$ with associated eigenvalues $\{\alpha_{l,i,j}\}$. We rewrite
\[ ... = \sum_i \lambda_i F(\tilde{\rho}_{l,i}, \ketbra{\psi_m}) = F(\tilde{\rho}_l, \ketbra{\psi_m})) \leq \epsilon(n), \]
where $\tilde{\rho}_l$ is the marginal of $ T_{l, w, x, z} ... T_{1, w, x, z} \ket{0}^{\otimes r}$ on the 1st $n$ qubits, and $\epsilon(n) < 1/2^{\alpha n}$ for arbitrary $\alpha \in (0,1)$.

We may now apply Lemma \ref{indist} with the unitary sequence $T_{k, w, x, z}, ..., T_{1, w, x, z}$ and $U_{l_1}, ..., U_{l_k}$, replacing $n$ in the Lemma by $2n$, and assuming that $U_{m_1}, ..., U_{m_k}$ all mark states that would be hard to prepare in $k$ steps with fidelity greater than $\epsilon(n)$. The rest of the proof would then follow as does the end of Theorem \ref{mqst} after equation \eqref{eq:nb}, via the operational interpretation, convexity, and monotonicty of the 1-norm.

The one problem is that Arthur may apply $U_{l}$ for $l$ sufficiently small that the marked state for this oracle is approximable with circuits of length $k$. We then note that by running the preparation circuit in reverse, running a simple circuit that marks the state $\ket{0}^{\otimes l}$, and then running the preparation circuit again, we can emulate such an oracle in at most $3 k$ gates with no oracle calls. Replacing these oracle calls by said emulation, we replace all oracle calls with marked states aproximable with larger than $\epsilon(n)$ fidelity using $4 k$ gates by up to $3 k$ gates each. Substituting $k \rightarrow 4 k$, we may assume that no states are prepared that resemble marked states for oracle calls, including $\ket{\psi_n}$, with precision larger than $\epsilon(n)$, and we apply Lemma \ref{indist}.

\end{proof}

\begin{proof}[(Proof of Corollary \ref{mqsl} part, Remark \ref{noquery})]
The 2nd part of the Theorem follow from allowing Merlin or the advice to send the marked state.

The Remark follows from Merlin specifying the index of the marked state, or from Arthur having access to a known set of possible marked states, as well as the parameters passed to $STATE\_DIAG$. Arthur may then run the $STATE\_DIAG$ program to generate the marked state, and verify it with the oracle.
\end{proof}

\begin{prop} \label{simplybqpnp}
Let $(x_n)_{n=1}^\infty$ be a family of arbitrarily chosen binary strings, $L$ a binary language, and $(\orac_n)_{n=1}^\infty$ be a family of standard oracles on states of $2n + 1$ qubits each defined for an $n$-length computational basis string state $\ket{y}$ and qubit state $\ket{\phi}$ by
\begin{enumerate}
	\item $\orac_n \ket{x_n} \otimes \ket{y} \otimes \ket{\phi} = \ket{x_n} \otimes \ket{y} \otimes X \ket{\phi}$ if $y \in L$.
	\item $\orac_n \ket{z} \otimes \ket{y} \otimes \ket{\phi} = \ket{x_n} \otimes \ket{y} \otimes  \ket{\phi}$ for any $n$-bit string $z \neq x_n$ or if $z = x_n$ but $y \not \in L$.
\end{enumerate}
Then $NP^{(\orac_n)} \not \subset BQP^{(\orac_n)}$, and $P^{(\orac_n)}/poly \not \subset BQP^{(\orac_n)}$. Let $(U_n)_{n=1}^\infty$ be a family of quantum oracles defined by
\[ U_n := (H_n \otimes \id_{n+1}) \orac_n (H_n \otimes \id_{n+1}) , \]
where $\id_{n+1}$ is the $(n+1)$-qubit identity, and $H_n$ the $n$-qubit Hadamard transform. Then $QCMA^{(U_n)} \not \subset BQP^{(U_n)}$.
\end{prop}
\begin{proof}
For each $n$ and input string $y$, $U_n$ is essentially a Grover oracle. The optimality of Grover search \cite{zalka_grovers_1999} implies that with access to $U_n$, it takes at least $O(2^{n/2})$ oracle queries to find the marked state given a fixed $y \in L$ or to determine whether $x \otimes y$ is a marked string for arbitrary $x$ and $y$. Given a copy of the marked string, trusted or not, a classical algorithm can use the standard oracle to verify an untrusted copy of $x_n$ as witness that $y \in L$.

The same bound applies to $BQP^{(\tilde{U}_n)}$. In this case, the oracle is no longer callable as a standard oracle.
\end{proof}

\section{Physical Realizability} \label{sec:easy}
As in \cite[p27]{aaronson_complexity_2016} one may construct for any $r$ qubits a bipartite quantum system of $2 r$ qubits initially prepared in the state $\ket{0}^{\otimes r}$. One then applies a random circuit of polynomial size, and measures the first $n$ qubits in the computational basis, yielding a state of the form $c \otimes m \otimes \ket{\psi_{c,m}}$, where $m$ is a classical variable storing the measurement result, and $c$ is a classical variable storing the random circuit applied. As noted in \cite{aaronson_complexity_2016}, $\ket{\psi_{c,m}}$ is thought to have complexity $2^{\Omega(n)}$.

\begin{proof}[(Proof of Theorem \ref{mqphys})]
Let $L$ be an arbitrary, unary language. Let $N = 2^n$, and $\kappa \in (0,1)$. We construct a classical-quantum channel $ \Phi_{L, \kappa, (c_1,m_1), ..., (c_n,m_n)} : \Hil_{N}^{\otimes n} \otimes \Hil_2 \rightarrow \Hil_2$
formally a completely positive, trace-preserving, linear map that extends via linearity to densities. When clear from context, we denote $\Phi := \Phi_{L, \kappa, (c_1,m_1), ..., (c_n,m_n)}$. $\Phi$ first performs a swap test on the 1st $n$ qubits of its input, which yields a random bit if $\ket{\phi} \perp \ket{\psi_{c,m}}$, and yields the value 1 if $\ket{\phi} = \ket{\psi_{c,m}}$. If at least $\kappa n$ of the swap tests yield the value 1 and $n \in L$, then $\Phi$ applies a Pauli $X$ gate to its final input qubit. In either case, $\Phi$'s output is the quantum state of the final input register. For notational convenience, let
\[ \ket{\psi} := \ket{\psi_{c_1,m_1}} \otimes ... \otimes \ket{\psi_{c_n, m_n}} \]
when it is clear from context or not important what $(c_1, m_1), ..., (c_n, m_n)$ should be.

$\Phi$ acts analogously to the quantum oracle. We see immediately that if $n \in L$ and is passed to $\Phi$, then $\Phi$ acts as a NOT gate on the final qubit. If a string of orthogonal states is passed as input or if $n \not \in L$, then $\Phi$ acts as identity. It is clear that with a copy of $\ket{\psi}$, one query to $\Phi$ is sufficient to decide $L$.

To recover the lower bound on quantum circuit complexity without a provided quantum witness or advice state, we note $\Phi$ to be less powerful than access to the quantum oracle $U$ that is the identity for a given $n$ if and only if $n$ is in the language, and when $U$ is not the identity, it inverts the phase $\ket{\psi}$ and acts as the identity on orthogonal states. Given such a $U$ and an $n^2 + 1$-qubit state $\ket{\sigma}$, we approximate $\Phi$ by the following procedure:
\begin{enumerate}
	\item Construct the extended state $\ket{\sigma} \otimes \ket{0}^{\otimes n^2} \otimes \ket{+}$, where $\ket{+} = (\ket{0} + \ket{1})/\sqrt{2}$. 
	\item Let $SWAP$ denote the operation that swaps qubits at indices $0...n^2$ qubits with qubits at $n^2+2...2n^2 + 2$. Apply SWAP conditionally (equivalent to $n^2$ conditional swap gates) on the last qubit to obtain the state
	\[ \frac{1}{\sqrt{2}} (\ket{\sigma} \otimes \ket{0}^{\otimes n^2} \otimes \ket{0} + SWAP(\ket{\sigma} \otimes \ket{0}^{\otimes n^2} \otimes \ket{1})) . \]
	\item Apply $U$ to the 1st $n^2$ qubits. If $\ket{\sigma} = \ket{\psi} \otimes \ket{\phi}$ for some qubit state $\ket{\phi}$, then we obtain the state
	\[ \frac{1}{\sqrt{2}} \big ( (-1)^{L(n)} \ket{\psi} \otimes \ket{\phi} \otimes \ket{0}^{\otimes n^2} \otimes \ket{0} 
		+ ((-1)^{L(n)} \braket{\psi | 0...0} \ket{\psi} + \braket{\eta | 0...0} \ket{\eta})  \otimes \ket{\phi} \otimes \ket{\psi} \otimes \ket{1} \big ) . \]
		for some $\ket{\eta} \perp \ket{\psi}$. If $\ket{\sigma} \perp \ket{\psi}$, then we obtain
		\[ \frac{1}{\sqrt{2}} \big ( \ket{\sigma} \otimes \ket{0}^{\otimes n^2} \otimes \ket{0} 
		+ ((-1)^{L(n)} \braket{\psi | 0...0} \ket{\psi} + \braket{\eta | 0...0} \ket{\eta}) \otimes \ket{\psi} \otimes \ket{1} \big ) . \]
		That $U$ marks the $\ket{0...0}$ state is somewhat unintentional. We see from Lemma \ref{manystates} that with probability at least $1 - O(1/2^{n/2})$, $\braket{\psi | 0...0} \leq 2^{-n/2}$. Hence with this small error, we may assume that $\braket{\psi | 0...0} = 0$, and $\ket{\eta} = \ket{0...0}$.
		
		The form of the state when $\ket{\sigma}$ is neither equal nor orthogonal to $\ket{\psi}$ follows from linearity.
	\item Again apply SWAP conditionally (equivalent to $n^2$ conditional swap gates) on the last qubit. Trace out the qubits at indices $n^2+2...2n^2 + 2$. Apply the Hadamard gate to the last qubit, obtaining when $\ket{\sigma} = \ket{\psi}$ the state
	$\ket{\psi} \otimes \ket{\phi} \otimes \ket{1 + L(n)}$, where the in-ket addition is modulo 2, and when $\ket{\sigma} \perp \ket{\psi}$, $\ket{\sigma} \otimes \ket{0}$.
	\item Apply a controlled not gate, with the last qubit as the control and index $n^2 + 1$ as target.
	\item Trace out all but the qubit at index $n^2 + 1$.
\end{enumerate}
The reason to use $n$ individual marked states rather than just 1 is so that when $\kappa > 1/2$, the probability of $\Phi$ acting non-trivially for an input having low overlap with the marked state decays exponentially with $n$. One could input an entangled state, but this will not result in higher probability of success. The parameter $\kappa$ allows the swap test to be approximate, which is important for physical realizability.

There are few notable differences from the proof of Theorem \ref{mqsl}. First, we must eliminate trusted, classical advice, which can simply tell the answer for a unary language. Hence we lose the separations with polynomial classical advice. Second, we no longer use the program $STATE\_DIAG$, but rely on the assumption that since $\ket{\psi_{c,m}}$ is probably exponentially hard to prepare, $n$ independent states of this form are at least as hard. The untrusted classical witness and input string are insufficient to distinguish whether $n \in L$, which is determined arbitrarily. The rest of the argument then follows as would Theorems \ref{mqst} and \ref{mqsl}.
\end{proof}
To emphasize the role of state complexity, the A \& K counting argument would not immediately yield Theorem \ref{mqphys}, because the polynomially encodable pair of classical variables $c \otimes m$ fully determine the marked state.

\begin{rem}
Formally, we must make the classical configurations $(c_1, m_1), ..., (c_n, m_n)$ accessible at least to the prover. This is not a problem when the channel $\Phi$ is hypothetically usable an arbitrary number of times by the prover, which is fully consistent with the theoretical separation, but physically may require an exponentially hard step to regenerate the (spent) state.

One alternative is to extend the language to $\tilde{L}$, which includes strings of the form $\ket{0...0}^{n- \log n} \otimes \text{SEP}  \otimes c \otimes \text{SEP} \otimes m$, where SEP is a special character used to separate parts of the input. Another is to allow the verifier unlimited access to the oracle. We may construct for each channel oracle a corresponding standard oracle that reveals the needed classical bits and is repeatable, and assume that this oracle is accessible whenever $\Phi$ is once-accessible.
\end{rem}

The constructions here show parallels with those of \cite{harlow_quantum_2013}. In their formulation, a black hole would be modeled as a polynomial-length random circuit. After matter is input to form the black hole, effectively random bits are re-emitted as Hawking radiation. An observer who tries to reconstruct the internal state of the black hole with knowledge of the physical process (and hence effective random circuit), input matter, and measured state of Hawking radiation may face a situation analogous to that of Theorem \ref{mqphys}.

\subsection{Experimental Realizability in Distributed Quantum Computing}
The fundamental assumption of the physically realizable channel oracle is that generating a complex marked state is possible with polynomial, random circuits and qubit measurements. Random circuit sampling is experimentally accessible \cite{arute_quantum_2019}. A quantum computer in the role of an oracle/server generates the marked state and associated classical string. After state generation, the server may publicly broadcast the classical information to a client in the role of Arthur, and a ``supercomputer" in the role of Merlin. The server further uses private, local randomness to determine whether the unary string of length $n$ is in the language. When Arthur receives $n$, the problem has been constructed. It is not necessary to solve the problem to demonstrate its physical existence.

If we do wish to solve the problem, Merlin should prepare another copy of the marked quantum state. Near-term quantum circuits lack error correction and are depth-limited, preventing a simulatable Merlin from running a purely quantum brute force search. Merlin may instead run the quantum circuit coded by $c$ a potentially exponential number of times as separate quantum computations, stopping when the measured qubits happen to match the given value of $m$. Alternatively, Merlin may attempt to trick Arthur by creating a different quantum state. One-way quantum communication from Merlin to Arthur transmits the state.

Arthur may then send Merlin's quantum state to the oracle, again using one-way quantum communication. We note that:
\begin{rem}
Though the oracle-like channel $\Phi$ is defined to perform a conditioned quantum operation in analogy with the quantum oracle, it might be easier in practice for $\Phi$ to return a classical bit, requiring only classical communication from the oracle to Arthur. This is simple to implement by replacing the final input qubit by the local state $\ket{0}$, then measuring what would be the output qubit and sending the classical result. It is not hard to see that equivalent separations in complexity result.
\end{rem}
\noindent Arthur's decision is essentially equivalent to returning the oracle's reply.

The protocol is sketched in Figure \ref{fig:prot}. The purpose of the parameter $\kappa$ is to allow for noise in the oracle and state transmission, swap tests, and computation.

\begin{figure}[htb!] \scriptsize \centering
	\includegraphics[width=0.5\textwidth]{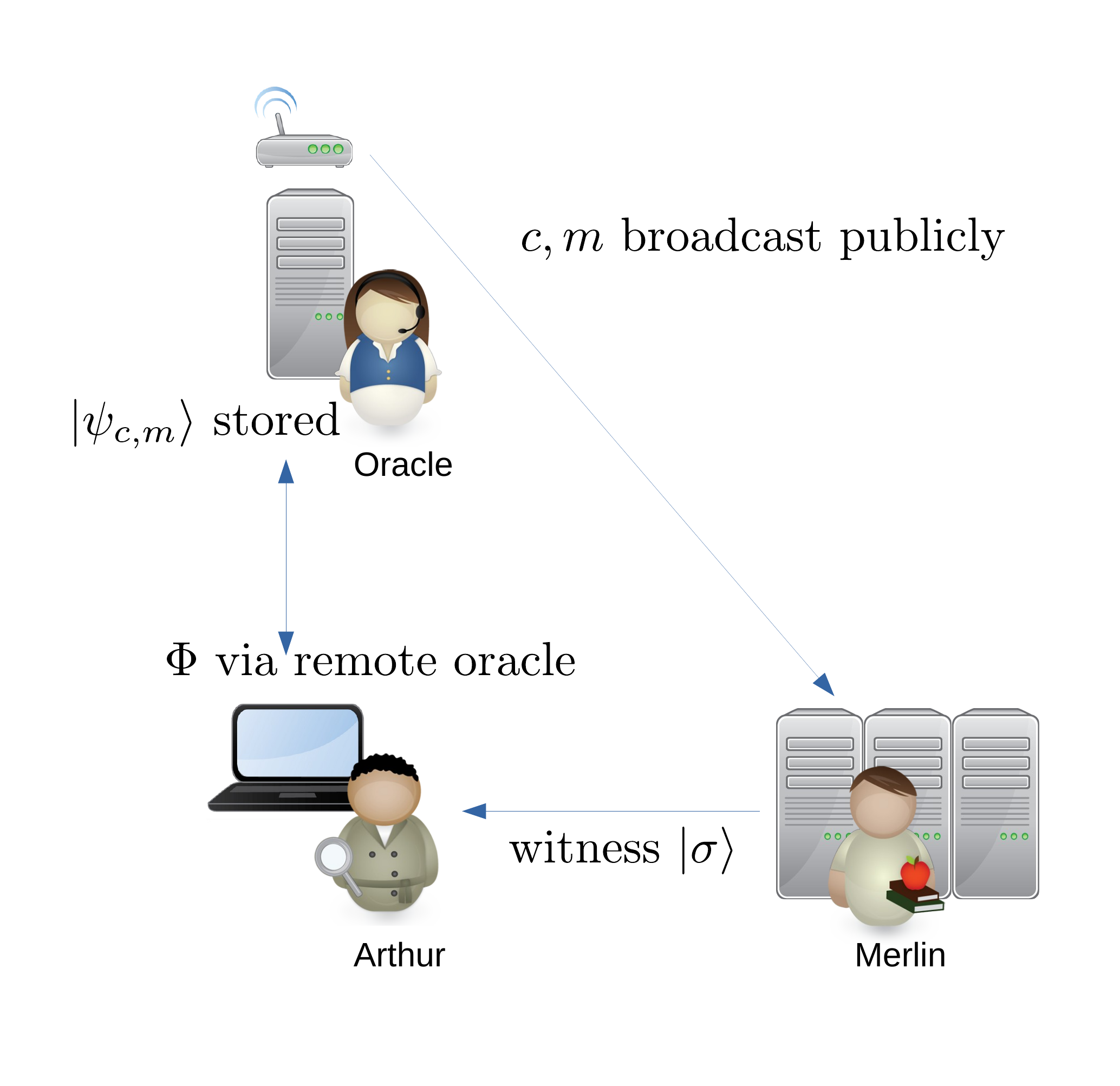}
	\caption{A diagram of the protocol to experimentally demonstrate quantum advantage of witness states. \label{fig:prot}}
\end{figure}

It is important for the interpretation of this problem that $\Phi$ be an effective black box, as we propose to enforce by physical separation. Communication between quantum computers is an emerging technology, and short-range experiments show promise to transfer single qubits between small quantum computers \cite{magnard_microwave_2020} faster than the decoherence times for well-optimized processing qubits \cite{jurcevic_demonstration_2020}. Importantly, our protocol only requires one-time, one-way quantum communication from Merlin to Arthur, and then from Arthur to the oracle. We thereby avoid the overhead of waiting for links to reset.

To solve the protocol with an honest Merlin or to check that Arthur cannot solve the protocol without quantum assistance requires exponential computation time, made somewhat more tractable by splitting into many runs of polynomial length quantum computations. Though not ideal, this cost is conventionally acceptable. Hard classical verification is standard in random circuit sampling demonstrations \cite{arute_quantum_2019, bouland_complexity_2019, pednault_leveraging_2019}. Often these experiments aim for an intermediate regime, in which classical supercomputers may repeat a problem using orders of magnitude more time and space than would a quantum computer.

The quantum witness advantage demonstration creates a scheme for proof of work, allowing one to verify that an untrusted party has access not only to a quantum computer, but to a large number of cycles (assuming the aforementioned separation of random vs. deterministic state complexity).

\input{tsir_oracle}

\section{Conclusions \& Outlook}
In complexity theory, one cannot necessarily replace a quantum proof or advice state by a polynomial classical description, yet it is also not obvious that a quantum proof or advice state has an exponential advantage in solving a classical problem. It is contrastingly easy to construct a fundamentally quantum task that is trivial given a particular quantum state, but for which classical assistance is not much use.

The quantum oracle construction of Aaronson \& Kuperberg \cite{aaronson_quantum_2006} adds some quantum elements to a classical complexity class, requiring quantum information to code a verifiable proof. We show an analogous construction in which quantum information isn't necessary to specify the proof, but a quantum proof is necessary to efficiently query the oracle. Consequently, even much larger classical resources are of limited power in solving the problem. That extra classical information has little effect may suggest that this construction is more analogous to making the problem in certain ways quantum than it is to a classical oracle. Nonetheless, standard oracles are bizarrely useful in preparing quantum states \cite[prop 3.3.5]{aaronson_complexity_2016}, so it is still plausible that they could have a role in this setting.

To show QMA $\neq$ QCMA without an oracle, one strategy is to seek a property of quantum states that is unlikely to hold for easily preparable states, but is easy to check for a given state. A tempting path is to ask Arthur to test the output of a polynomial circuit on a particular subset of hard-to-prepare states. Unfortunately, Arthur also can't efficiently verify that a state given by Merlin is hard or corresponds to the correct index, so this problem is probably neither in QMA nor in QCMA. If a property is sufficiently rare and sufficiently independent from the hardness of state preparation, it might nonetheless be possible to show that with non-zero probability, an instance will have only hard-to-prepare proofs. The canonical, existing candidate property is that the state be identified by a local Hamiltonian, e.g. by being a ground state. The challenge in this case would be to show that any such states are hard to prepare from any classical description.

Experimental demonstrations of quantum witness advantage should be around the corner, assuming the hardness of deterministically preparing certain randomly generated quantum states. This possibility motivates further work on the conjectured separation between random and deterministic state complexity. State complexity is fundamentally connected to the holistic nature of highly-entangled quantum systems, which underpin error correction and holographic spacetime \cite{almheiri_bulk_2015}. It is possible that a full resolution to QMA vs. QCMA will not only be a fundamental result on the power of quantum computational resources, but also reveal new physics.

\section{Acknowledgements}
I thank Bill Fefferman for his guidance on this project.
{ \scriptsize

\bibliography{states}
\bibliographystyle{unsrt}

}

\end{document}

%% file: intro_2.tex
\section{Introduction}

\newcommand{\orac}{\mathcal{O}}

In the black box or oracle model, a computation has access to a subroutine with $O(1)$ cost but no knowledge of its implementation. The oracle may perform an otherwise costly computation in $O(1)$ time or reveal information that would otherwise not be discoverable from the inputs. A quantum computer may conventionally use a classically-defined oracle in the computational basis, calling from superposition. Quantum-classical oracle separation is at the root of quantum complexity theory, going back to some of the earliest known separations and lower bounds \cite{bernstein_quantum_1997, boyer_tight_1998, zalka_grovers_1999}. An oracle function with the form $\orac : \{0,1\}^n \rightarrow \{0,1\}$ is known as a \textit{standard oracle} and has a natural extension from classical bit strings to quantum states in the computational basis.

In the quantum setting, there are however more ways in which one could define a sensible notion of an oracle. Aaronson and Kuperberg introduce the notion of a \textit{quantum oracle} \cite{aaronson_quantum_2006}, an arbitrary unitary callable with $O(1)$ cost as a black box subroutine. Aaronson and Kuperberg's quantum oracle has no obvious classical analog, relying on states in arbitrary bases. Fefferman and Kimmel construct an ``in-place" quantum oracle that appears to possess a classical analog but is still not a standard oracle, as it is not self-inverse \cite{fefferman_quantum_2018}. One may consider other forms of potentially restricted, non-classical oracles, such as quantum channels that may introduce mixture, or quantum oracle families with promises on the unitaries available.

Quantum or classical, oracles go hand-in-hand with diagonalization techniques, a notion formalized for subset-based oracles in \cite[Section 5, Recipe 1]{fefferman_quantum_2018}. Intuitively, one shows that each possible program using polynomial oracle calls must in some case fail to yield a sufficiently distinguishable final state on a string in the language from one not in the language. This approach has an information-theoretic interpretation: the information any fixed program obtains from polynomial oracle calls does not scale as quickly as the information needed to accurately decide the problem. In \cite{aaronson_quantum_2006}, that information is what's needed to unlock a secret hidden in the quantum oracle. In particular, it cannot solve the following task:
\begin{definition}[Marked Quantum State Task (MQST)] \label{def:mqst}
Given $n \in \mathbb{N}$ and $\ket{\psi} \in \Hil_{2^n}$ ($n$-qubit Hilbert space), let $U : \Hil_{2^n} \rightarrow \Hil_{2^n}$ be a unitary (quantum oracle) with the promise that either:
\begin{enumerate}
	\item $U \ket{\psi} = - \ket{\psi}$, and $U$ acts as identity on orthogonal states.
	\item $U$ is the identity.
\end{enumerate}
The task is to distinguish these cases with at least $1/poly(n)$ success probability.
\end{definition}
Given an untrusted copy of the marked state, it is easy to verify case (1). Most quantum states on $n$ qubits lack a polynomial-length classical description within reasonable orders of approximation (see Section \ref{sec:statecomplexity}, or the original explanation of \cite{aaronson_quantum_2006}). Hence any scheme to reveal the marked state using polynomially-many bits of classical information must in some cases fail to distinguish superpolynomially large sets of candidates.

As a first result in this paper, we show that the information-theoretic barrier and diagonalization are not necessary for the primary results achieved in \cite{aaronson_quantum_2006}. To do so, we find exponentially large families of quantum states that have polynomial classical descriptions but are not approximable using any polynomial-size circuit. We call the \textit{state complexity} of a quantum state the minimum circuit length needed to prepare it from the initial state $\ket{0}^{O(poly(n))}$ using $O(poly(n))$ elementary gates.
\begin{theorem} \label{thm:main}
There exist families of marked quantum state oracles $(U_n)_{n=1}^\infty$ for which case (1) of MQST can be verified using an untrusted, polynomial-size, classical string as proof and polynomially many queries to $U_n$ but needing exponential time. Given an untrusted, polynomial-size quantum state as proof, case (1) can be verified in polynomial time with one query. The marked states are specified to $1/poly(n)$ precision by polynomial, classical bit strings.

For some such families of quantum oracles, polynomial compute time and polynomial queries suffice to verify case (1) of MQST given a trusted, polynomial quantum state as proof that depends on the dimension of the marked state but not otherwise on the marked state.
\end{theorem}
We replace the query complexity separation of \cite{aaronson_quantum_2006} by a separation in state complexity. Given the the marked state's classical description, it is possible to construct the marked state using an exponentially large quantum circuit and verify case (1) of MQST using one query. The surprising aspect of Theorem \ref{thm:main} is that the usual oracle and query complexity machinery is largely bypassed. We may even restrict the marked state to a family that is deterministically known to the program as a function of the input size, but because of the state's complexity, such marked states may still be impossible to efficiently verify without an externally-provided quantum witness.

Variants of MQST are convert to decision problems to show separations in computational complexity. We recall the complexity class Quantum Merlin Arthur (QMA), in which a quantum, polynomially-bounded verifier checks a polynomial-size quantum witness state. We recall the class Quantum-Classical Merlin Arthur (QCMA), in which the witness state is constrained to be a classical bitstring. We also recall the notion of advice, which is a classical string or quantum state that depends on the input size but not otherwise on the input string. For a complexity class $C$, we denote by $C/log$ and $C/poly$ the modified class respectively with logarithmic or polynomial advice. By $C/qlog$ and $C/qpoly$ we denote the respective classes with quantum advice. By $C^{(U_n)}$ we denote the modified class with access to a family of quantum oracles $(U_n)$, where $n$ is the input size.

In \cite{aaronson_quantum_2006}, it is shown that $QMA^{(U_n)} \neq QCMA^{(U_n)}$, and $BQP^{(U_n)}/poly \neq BQP^{(U_n)}/qpoly$ for families of quantum oracles $(U_n)$ as in MQST. Unsurprisingly, we are able to combine these separations. More surprisingly, we do so bypassing the query separation argument.
\begin{cor} \label{cor:complexity}
There exists a quantum oracle family $(U_n)_{n=1}^\infty$ for which
\begin{equation*}
BQP^{(U_n)}/qpoly \not \subseteq QCMA^{(U_n)}/poly \text{, and } QMA^{(U_n)} \not \subseteq QCMA^{(U_n)}/poly ,
\end{equation*}
though one query is sufficient to verify a classical witness given exponential computation time.
\end{cor}
The technical version appears as Corollary \ref{mqsl}. When we replace the oracle herein by a standard oracle, we find that $BQP$ does not contain $NP$ or $P/poly$ relative to this oracle.
\begin{prop} \label{npbqpintro}
Let $(\orac_n)_{n=1}^\infty$ be a family of standard oracles with marked states in the computational basis. Then $NP^{(\orac_n)} \not \subset BQP^{(\orac_n)}$, and $P^{(\orac_n)}/poly \not \subset BQP^{(\orac_n)}$. Let $(U_n)_{n=1}^\infty$ be a family of quantum oracles with marked states as binary strings in the Hadamard basis. Then $QCMA^{(U_n)} \not \subset BQP^{(U_n)}$.
\end{prop}
The technical version of this Proposition appears as \ref{simplybqpnp}. The power of NP relative to this oracle and its replacement by QCMA when the oracle is transformed to the Hadamard basis illustrates the barrier to replacing the quantum oracle by a standard oracle. More substantially, we may ask why we cannot replace the complex state by, for instance, an efficient polynomial design (see for example \cite{harrow_random_2009}). The primary barrier here is that whenever the state admits an efficient quantum circuit, a polynomial quantum computation suffices to generate the marked state from that string. Hence schemes based on classical, hard-to-discover keys/seeds may fail to separate QMA from QCMA, as a classical prover can reveal the secret that cracks the scheme. Though it would be easy to show that NP does not contain QMA or QCMA relative to a quantum oracle, this separation is arguably trivial and unfair without a way for classical programs to access information hidden in the quantum oracle.

Though it is by definition hard to prepare a specified complex state, \cite[p27]{aaronson_complexity_2016} notes that it may not be hard to generate random states that are typically complex. Assuming that the procedure therein does generate a random, complex state, we propose an experimental demonstration of the quantum oracle complexity separation. For this, we need to relax the oracle definition further, allowing a once-callable quantum channel to take the place of the oracle. Rather than inserting a specified unitary at arbitrary points in computation, we will allow the computation to insert one call to the quantum channel with assume $O(1)$ cost. Then:
\begin{theorem} \label{mqphys}
There exists a quantum channel oracle $\Phi$ for which $QMA^{\Phi} \neq QCMA^{\Phi}$.

Assume that the complex state generation procedure as suggested in \cite[p27]{aaronson_complexity_2016} indeed generates complex states. Then polynomial quantum space and time suffice to:
\begin{enumerate}
	\item Generate a marked state and quantum channel oracle $\Phi$ such that $QMA^{\Phi} \neq QCMA^{\Phi}$ with high probability, and where the marked state for $\Phi$ is classically specified in polynomially many bits.
	\item Make the marked state's classical description known to the prover and verifier.
	\item Apply $\Phi$ once when called by the verifier.
\end{enumerate}
\end{theorem}
The once-callable aspect of $\Phi$ technically arises from the fact that one copy of the specific marked state is created, and this copy is consumed in the oracle call. An exponentially more powerful prover can generate as many copies as desired from the obtained classical description. This channel oracle could be used in a quantum proof of work scheme, in which a polynomial quantum computer generates instances that probably require exponentially stronger, quantum quantum resources to solve.

Counter-intuitively, a standard oracle may spoil the secret of a quantum oracle as above. As shown in \cite[Proposition 3.3.5]{aaronson_complexity_2016}, one may prepare arbitrary states using a standard oracle that directly reveals bits and phases of the marked state. Although the standard oracle is defined by a classical function, it may be called in superposition, enabling a procedure that builds up an exponential number of amplitudes directly, without ever having to obtain the exponentially long bit string needed to specify the state classically.

The techniques of \cite{kretschmer_quantum_2020} use a state preparation oracle model to bound a quantum computer's ability to generate heavily-weighted output strings of a quantum state. The analogous task without oracles is believed hard for classical computers \cite{aaronson_complexity-theoretic_2017, boixo_characterizing_2018}. In the formulation presented there, the use of a quantum oracle precludes a classical analog. We replace this oracle by a standard oracle, showing separation between classically stochastic and quantum use analogous to that in oracle-free models. As in MQST, superiority of quantum resources arises not from the information required to specify arbitrary quantum states, but from hardness of preparing states. To do this, we define a rotated version of the heavy output sampling problem, using classical inaccessability of a complementary basis as in Proposition \ref{npbqpintro}, but without needing an explicit, quantum oracle:
\begin{theorem}[Classical Hardness of RXHOG] \label{rxhogclassintro}
Let a classical algorithm attempt to solve the sampling problem RXHOG as in Definition \ref{rxhog} with access to polynomially many calls to a standard state preparation oracle, and a table of coefficients of a quantum state that defines the sampling problem. Such an algorithm achieves a value no higher than $1/2^n + O(poly(n))/2^{2n}$, whereas a quantum algorithm may achive a value of $2/2^n$.
\end{theorem}
The technical version appears as Theorem \ref{rxhogclass}.

Ultimately, if QMA $\neq$ QCMA without an oracle, then there must exist low-energy states of local, gapped Hamiltonians that are polynomially specified (by the Hamiltonian's description) but not preparable by polynomial circuits. Reformulating quantum oracle separation in terms of state complexity, rather than state codeability, brings it closer to the oracle-free setting. We have already seen that there exist states with simple classical descriptions that are hard to prepare. To separate QMA from QCMA, it would suffice to find a polynomial time routine for a class of such states that verifies if a provided such state matches a classical description, in which case the language of classical descriptions would fall in QMA and not QCMA. In this work, we see that the presence of quantum oracle can provide such verification, while a standard oracle can break the preparation hardness.

In Section \ref{sec:hardstate}, we review known results on state complexity and define the hard state preparation routine that allows us to efficiently index some hard states. In Section \ref{sec:mqst}, we show Theorem \ref{thm:main} without directly invoking complexity classes. In Section \ref{sec:complexity}, we show Corollary \ref{cor:complexity} and Proposition \ref{npbqpintro}, discussing the implications of marked states for complexity classes relative to quantum oracles. In Section \ref{sec:easy}, we show Theorem \ref{mqphys} and discuss possible experimental implementations. In Section \ref{sec:tsir}, we show Theorem \ref{rxhogclassintro} and discuss standard state preparation oracles in the context of sampling problems.

\subsection{Basic Notations}
For a pair of pure states $\ket{\psi}, \ket{\phi}$, recall the overlap be defined via the inner product as $|\braket{\psi | \phi}|^2$. We denote by $F(\cdot, \cdot)$ the (squared) fidelity between two quantum states, given by $F(\rho, \sigma) = \|\sqrt{\rho} \sqrt{\sigma}\|_1^2$ in trace norm. For pure states, fidelity is equal to the overlap. For a set $S$, we denote by $|S|$ its cardinality, and otherwise, $| \cdot |$ denotes the absolute value of a complex function. We use $\lceil \cdot \rceil$ to denote the ceiling function and $\lfloor \cdot \rfloor$ for the floor function.

%% file: tsir_oracle.tex
\newcommand{\NN}{\mathbb{N}}
\newcommand{\RR}{\mathbb{R}}
\newcommand{\EE}{\mathbb{E}}
\newcommand{\sep}{\square}

\section{Quantum Sampling Advantage with Oracles} \label{sec:tsir}

\begin{prob}[Rotated Linear Cross-Entropy Heavy Output Generation (RXHOG)] \label{rxhog}
Given a random unitary $U \in U(n)$ distributed by probability measure $\mu_n : U(n) \rightarrow [0,1]$, a unitary $u_n \in U(n)$ for each $n \in \NN$, and $f : \NN \rightarrow \RR^+$, output $k(n)$ distinct samples $(z_j \in \{0,1\}^n)_{j=1}^k$ such that $\EE_{U,j,r}[|\braket{z_{j, U}|u_n U | 0^n}|^2] \geq f(n)/2^n$, where $r$ denotes any randomness generated by or available to the algorithm.

XHOG is RXHOG with $u_n = \id$ for all $n$.
\end{prob}

\begin{theorem}[Classical Hardness of RXHOG] \label{rxhogclass}
Let a classical algorithm attempt to solve RXHOG with access to:
\begin{itemize}
	\item Time, space, and randomness bounded by any computable, finite function of input size.
	\item $P(n) = O(poly(n))$ calls to the oracle $O_\psi$, where $\ket{\psi} = U \ket{0^n}$ for each given $U$ as distributed according to the Haar measure on $n$-qubit unitaries.
	\item A complete table of coefficients of $\psi_Z$ in the computational basis to any precision allowable within space and time constraints.
\end{itemize}
For $u_n = H^{\otimes n}$ and such a classical algorithm,
\[\EE_{U,j}[|\braket{z_j|u_n U | 0^n}|^2] \leq 1/2^n + O(poly(n)/2^{2n}) . \]
\end{theorem}

\subsection{Formulating the Standard State Preparation Oracle}
The purpose of this section is to formalize \cite[Proposition 3.3.5]{aaronson_complexity_2016} with rigorous error bounds and a precisely defined oracle. A reader who is not interested in these details may skip this section and read the referenced proposition's proof for intuition. We require precision to analyze RXHOG in the presence of oracles. While calling the constructed oracle in superposition should prepare the quantum state with enough fidelity and flexibility for RXHOG, classical calls should not be able to reveal heavy outputs.

For any real number $\alpha \in [0,1]$, let $\alpha = 0.\alpha[0]\alpha[1]\alpha[2]...$ denote the binary digits of $\alpha$ following the decimal point. In this notation, we represent $1 = 0.1111...$ as a repeating decimal. Let $\alpha[:k]$ denote the 1st through $k$th binary digits of $\alpha$. Let $\alpha[:k]$ denote the bit string $\alpha_1 ... \alpha_k$. Let $\alpha\{k\} = 0.\alpha[0] ... \alpha[k]$ denote $\alpha$ taken to the precision of the $k$th binary digit. For example, if $\alpha = 0.1011$, then $\alpha\{1\} = 0.1$, $\alpha\{2\} = 0.10$, and $\alpha\{4\} = \alpha$ as real coefficients, while $\alpha[:1] = 1$, $\alpha[:2] = 10$, and $\alpha[:4] = 1011$.

For a binary string $b \in \{0,1\}^n$, we denote by $b[j:k] = b_j...b_k$ the $j$th through $k$th bits. By $\ket{b}$ we denote the $n$-qubit ket vector of $b$ in the computational basis. By $0.b[1:k] = 0.b_1...b_k$ we denote the real number given by binary digits of $b$ to the specified precision.

To denote strings in classical computation, we concatenate numbers. For example, we may write ``$s = 01\sep010$" to denote the 5-symbol string of a 2-bit binary number, a separarating symbol ``$\sep$", and a 3-bit binary number. Letters may denote strings, for instance we could write $a = 01, b = 010$ and have $s = a \sep b$. For a pair of binary strings of equal length, we denote by ``$\oplus$" the bitwise XOR. For quantum computations, we use braket notation with ``$\otimes$" denoting the tensor product.
\begin{definition}[Controlled Rotation/Phase]
For a binary number $b \in \{0,1\}^n$ be an $n$-bit binary number. Let $C_{\text{rot},n} \in U(n+1)$ denote the controlled rotation gate given by
\begin{equation*}
\begin{split}
C_{\text{rot},n} \ket{b} \otimes \ket{0} & = \ket{b} \otimes (\sqrt{1 - (0.b_1...b_n)^2} \ket{0} + 0.b_1...b_n \ket{1}) \\
C_{\text{rot},n} \ket{b} \otimes \ket{1} & = \ket{b} \otimes (\sqrt{1 - (0.b_1...b_n)^2} \ket{1} - 0.b_1...b_n \ket{0}) .
\end{split}
\end{equation*}
Let $C_{\text{ph},n} \in U_n$ denote the controlled phase gate given by
\[C_{\text{ph},n} \ket{b} = \exp(2 \pi i \times 0.b_1 ... b_n) \ket{b}. \]
\end{definition}
\begin{lemma} \label{crotph}
For given $n \in \NN$ and universal gate set $G$, $C_{\text{rot},n}$ and $C_{\text{ph},n}$ can be implemented to overlap error at most $1 - poly(n)/2^{2 n}$ each with $O(n^\kappa)$ elementary gates for some $\kappa \in \RR^+$.
\end{lemma}
\begin{proof}
Rather than try to approximate $n$-qubit or $n+1$-qubit unitaries, we will rewrite both gates as sequences of 2-qubit gates. This lets us apply results on universal gate approximations in constant dimension.

The controlled phase gate $C_{\text{ph},n}$ is simpler and contains intuition relevant to both desired gates. We split the $n$-qubit gate into $n$ single-qubit controlled phase gates. For $j \in 1...k$, Let $C_{\text{p}, j}$ be defined by 
\begin{equation*}
\begin{split}
C_{\text{p}, j} \ket{0}  =  \ket{0} \text{, }
C_{\text{p}, j} \ket{1}  = \exp(2 \pi i / 2^j) \ket{1} .
\end{split}
\end{equation*}
Applying each $j$th gate to the $j$th qubit, the total multiplying coefficient is $\exp(2 \pi i b / 2^n)$, corresponding to the desired phase, where here $b$ is interpreted as a little-endian integer. We must however also account for the imprecision resulting from constructing these qubit gates. Via efficient approximations, we can approximate any single-qubit gate to precision $\epsilon$ using $c \log^{\tilde{\kappa}}(\epsilon)$ base gates for constants $c$ and $\tilde{\kappa}$ (where $\tilde{\kappa}> 1$ may follow from the original gate set or the particular metric) as long as $G$ obeys certain assumptions \cite{harrow_efficient_2002}. For rounding to the desired precision, we will set $\epsilon = 1/n 2^{2 n}$, requiring $c n^2 (2 n)^{\tilde{\kappa}}$ elementary gates.

The $n+1$-qubit $C_{\text{rot},n}$ conditionally rotates a state in the $X-Z$ plane of the Bloch sphere representation of a qubit. Let a qubit in half of this plane be given by $\sqrt{1 - \beta^2} \ket{0} + \beta \ket{1}$ for $\beta \in [0,1]$, and let $\alpha = \arcsin(\beta) \in [0, \pi/2)$. For a given $\ket{b} \otimes \ket{\psi}$, let $\alpha = 0.b_1...b_n$. Because rotations are additive, we may define $C_{\text{rot},n}$ as a sequence of gates $C_{\text{r}, j} \in U(2)$ for $j \in 0...n-1$, each given by
\begin{equation*}
\begin{split}
C_{\text{r}, j} \ket{1} \otimes \ket{0} & = \ket{1} \otimes (\cos(\pi / 2^{j + 1}) \ket{0} + \sin(\pi / 2^{j + 1}) \ket{1}) \\
C_{\text{r}, j} \ket{1} \otimes \ket{1} & = \ket{1} \otimes (\cos(\pi / 2^{j + 1}) \ket{1} - \sin(\pi / 2^{j + 1}) \ket{0}) \\
C_{\text{r}, j} \ket{0} \otimes \ket{\sigma} & = \ket{0} \otimes \ket{\sigma} \text{ } \forall \ket{\sigma} \\
\end{split}
\end{equation*}
and applied to the $j$th qubit in $b$ and the target qubit. As with the phase gates, we will tolerate an error up to $1/n 2^{n+1}$ in each rotation to bound the total error at $1/2^{2 n}$ at the same gate cost up to differences in the constants.
\end{proof}

\begin{definition}[Standard State Preparation Oracle]
Let $\ket{\psi}$ be a given quantum state, $b \in \{0,1\}^n$, $a \in \{0,1\}$,
\[\alpha_{b_1...b_k} = 2 \arcsin(\sqrt{tr(\ketbra{b_1...b_k 1}\ketbra{\psi})}) / \pi \in [0,1] ,\]
and $\phi_b = \braket{b|\psi}/|2 \pi \braket{b|\psi} |$ or 0 if $\braket{b|\psi} = 0$. Let $p \in \NN$ be given in binary. Define $O_\psi$ by the following classical functions:
\begin{equation*}
O_\psi(b_1 ... b_k \sep p \sep a) = 
\begin{cases} 
      b_1...b_k \sep p \sep (a \oplus \alpha_{b_1...b_k}[p]) & 0 \leq k < n \\
      b_1 ... b_k \sep p \sep (a \oplus \phi_{b}[p]) & k = n .
\end{cases}
\end{equation*}
For any input string $s$ that does not fit the above format, let $O_\psi(s) = s$.
\end{definition}
$O_\psi$ is its own inverse. Though $O_\psi$ is defined classically, $O_\psi$ trivially extends to a quantum oracle on the basis $\{\ket{0}, \ket{1}, \ket{\sep}\}$. Hence, we define a quantum state preparation routine:
\begin{definition}[Quantum Oracle State Preparation] \label{qprepsubr} \normalfont
Let $O_\psi$ be a standard state prepaparation oracle and $\epsilon > 0$ a given precision. Let $p(\epsilon) = \lceil 1/\epsilon \rceil$, where the ceiling  is taken in binary. We define a procedure in the following steps:
\begin{enumerate}
	\item Start with $\ket{\eta}$ as a ket in the 1-dimensional Hilbert space of 0 qubits. Let $b = b_1 ... b_n$.
	\item Given $\ket{\eta} = \sum_{b_1...b_k} \beta_{b_1...b_k} \ket{b_1...b_k}$, define the following subroutine:
	\begin{enumerate}
		\item For each $p \in 0...p(\epsilon)$ as a binary string, call
		\[ O_\psi(\ket{\eta} \otimes \ket{\sep} \otimes \ket{p} \otimes \ket{\sep} \otimes \ket{0}), \]
		collecting the final bits into a $p(\epsilon)$-qubit state. Discard $\ket{p}$ and the separator characters, yielding
		\[\sum_{b_1...b_k} \beta_{b_1...b_k} \ket{b_1...b_k} \otimes \ket{\alpha_{b_1...b_k}[:p(\epsilon)]}.\]
		\item Append a qubit prepared in 0 and apply bitwise controlled rotations, yielding
		\begin{equation*}
		\begin{split}
		& \sum_{b_1...b_k} \beta_{b_1...b_k} \ket{b_1...b_k} \otimes \ket{\alpha_{b_1...b_k}[:p(\epsilon)]} \\
		 & \hspace{5mm} \otimes
			\Big (\sqrt{1 - \sin(\pi \alpha_{b_1...b_k}\{p(\epsilon)\} / 2)}\ket{0}
			+ \sin(\pi \alpha_{b_1...b_k}\{p(\epsilon)\} / 2) \ket{1} \Big ) .
		\end{split}
		\end{equation*}
		\item Apply $O_\psi$ inversely to eliminate the explicitly stored value of $\alpha_{b_1...b_k}$, yielding
		\[ \sum_{b_1...b_k} \beta_{b_1...b_k} \ket{b_1...b_k} \otimes 
			\Big (\sqrt{1 - \sin(\pi \alpha_{b_1...b_k}\{p(\epsilon)\} / 2)}\ket{0}
			+ \sin(\pi \alpha_{b_1...b_k}\{p(\epsilon)\} / 2) \ket{1} \Big ) . \]
		Redefine $\ket{\eta}$ as this state.
	\end{enumerate}
	Apply the subroutine $n$ times for $k \in 0...n-1$, yielding an $n$-qubit $\ket{\eta}$.
	\item Let $\sum_{b} \beta_{b} \ket{b} = \eta$. For each $j \in 1...n$, apply the following subroutine:
	\begin{enumerate}
		\item For each $p \in 1...p(\epsilon)$, call
		\[ O_\psi(\ket{\eta} \otimes \ket{\sep} \otimes \ket{p} ) ,\]
		colleting the final bits into a $p(\epsilon)$-qbit state. Discard $\ket{p}$ and the separator characters, yielding
		\[ \sum_{b} \beta_{b} \ket{b} \otimes \ket{\phi_b[:p(\epsilon)]} . \]
		\item Apply bitwise controlled phase, yielding
		\[ \sum_{b} \beta_{b} e^{2 \pi \phi\{p(\epsilon)\} i} \ket{b} \otimes \ket{\phi_b[:p(\epsilon)]} . \]
		\item Apply $O_\psi$ inversely to eliminate the explicitly stored phase, yielding
		\[ \sum_{b} \beta_{b} e^{2 \pi \phi\{p(\epsilon)\} i} \ket{b} . \]
		Redefine $\ket{\eta}$ as this state.
	\end{enumerate}
\end{enumerate}
Let $\eta$ following this procedure be the final state.
\end{definition}

\begin{rem}
A classical, randomized algorithm with access to $O_\psi$ can prepare the mixture resulting from a computational basis measurement of $\ket{\psi}$ in much the same way that a quantum algorithm prepares the quantum state. Hence a classical algorithm can effectively sample in the computational basis, but not in the rotated basis. This motivates the rotation in the definition of the RXHOG.
\end{rem}

\begin{lemma}
For an $n$-qubit state $\ket{\psi}$, polynomial quantum queries to $O_\psi$ and elementary gates suffice to prepare an $n$-qubit state $\ket{\tilde{\psi}}$ such that $|\braket{\tilde{\psi}|\psi}|^2 \geq 1 - O(poly(n)/2^n)$.

Polynomial classical queries and elementary, classical operations suffice to prepare the computational basis-diagonal density $\tilde{\psi}_Z$ such that $F(\tilde{\psi}_Z, \psi) \geq 1 - O(poly(n)/2^n)$.
\end{lemma}
\begin{proof}
By Lemma \ref{crotph}, each quantum $C_{rot, n}$ or $C_{ph, n}$ gate yields adds an overlap error of at most $O(poly(n)/2^{2^n})$. This error is insignificant compared to the error of $1/2^n$ from using finite precision to prepare a state with given amplitude and phase coefficients from polynomial queries. Since the subroutine of Definition \ref{qprepsubr} uses polynomial amplitude and phase preparation steps, it induces an error of $O(poly(n)/2^n)$.
\end{proof}

\subsection{Classical Hardness of RXHOG via Standard State Preparation Oracle}

\begin{proof}[(Proof of Theorem \ref{rxhogclass})]
Given the information available to it including the oracle and randomness, the classical algorithm returns a sequence of strings $(z_j)_{j=1}^k$. Let $z$ be any of these strings for a given configuration. We will prove the Theorem by averaging over possible values of $\ket{\psi}$ that would yield the same information, first assuming that the algorithm is deterministic given this info. We will then use convexity to extend the bound to randomized algorithms.

In particular, the information given can at most reveal polynomial bits of polynomially many phases in the $X$ basis, which we may replace by assuming that the algorithm has access to polynomially many phases to arbitrarily high precision. Let
\[ S_{ph} = \{l : \text{ phase of $\ket{l}$ is known } \}, \]
which we may assume have non-zero amplitude. Let
\[ \ket{\eta} = \frac{1}{\sqrt{\sum_l |\beta_l|^2}} \sum_{l} \beta_l \ket{l}
	\text{, } \rho = \frac{1}{\sum_m p_m} \sum_{m} p_m \ketbra{m} , \]
where $\beta_l$ is the amplitude of each term with known phase, and $p_m$ is the probability of each term such that $m \not \in S_{ph}$. Let $\beta = \Big ( \sum_l |\beta_l|^2 \Big )^{1/2}$, and note that $\sum_m p_m = 1 - \beta$.

Let $CPHASE_{m, \phi}$  be given by
\begin{equation}
\begin{split}
\begin{cases} 
      CPHASE_{m, \phi} \ket{m} = e^{2 \pi i \phi} \ket{m} & \\
      CPHASE_{m, \phi} \ket{l} = \ket{l} & l \neq m .
\end{cases}\end{split}
\end{equation}
We now use that for any $z$,
\begin{equation}
\begin{split}
& \EE_{\ket{\psi}}[|\braket{z|X^{\otimes H} | \psi}|^2 | \text{ all magnitudes, fixed phases for $l \in S_{ph}$ }] \\
& \text{ }	= \EE_{\phi_m : m \not \in S_{ph}} [ | \langle z | X^{\otimes H} \circ CPHASE_{m, \phi_m} |
		 \rangle |^2 ] \\
& \text{ } = tr \big (\ketbra{z} X^{\otimes H} (\beta \ketbra{\eta} + (1 - \beta) \rho) X^{\otimes H} \big ) \\
& \text{ } = \beta \big | \big \langle z \big | H^{\otimes n} \big | \eta \big \rangle \big |^2
	+ (1 - \beta) tr \big (\ketbra{z} H^{\otimes n} \rho X^{\otimes H} \big ) \\
& \text{ } = \beta \big | \big \langle z \big | H^{\otimes n} \big | \eta \big \rangle \big |^2
	+ (1 - \beta) 1/2^n ,
\end{split}
\end{equation}
where the last equality follows from the fact that $\rho$ is diagonal in the $Z$ basis, rotated into a state that is diagonal in the $X$ basis, and then compared to a $z$-basis string. By \cite[Fact 10]{kretschmer_quantum_2020}, $|\beta_l|^2 \leq O(n)/2^n$ for all $l \in S_{ph}$, so $\beta \leq |S_{ph}|O(n) / 2^n$. This is enough to bound the expectation by $1/2^n + O(poly(n))/2^n$ but not yet sufficient to achieve the desired bound.

To complete the Theorem, we must use the fact that $\ketbra{n}$ is highly concentrated in the $Z$ basis, so it can't be highly concentrated in the mutually unbiased $X$ basis. Explicitly writing out the Hadamard transform,
\begin{equation*}
H^{\otimes n} \ket{\eta} = \frac{1}{2^{n/2}} \sum_{l \in S_{ph}, x \in 0...2^{n-1}} (-1)^{x \cdot l} \beta_l \ket{x} ,
\end{equation*}
where $x \cdot l$ is the bitwise dot product of binary representations. Again by \cite[Fact 10]{kretschmer_quantum_2020}, $|\beta_l| \leq O(\sqrt{n})/2^{n/2}$ for any $l \in S_{ph}$. For any $z$, via the triangle inequality,
\[ |\braket{z | H^{\otimes n} |\eta}| \leq \frac{1}{2^n} \sum_{l \in S_{ph}} O(\sqrt{n})
	\leq \frac{O(poly(n))}{2^n} .\]
Squaring both sides completes the Theorem for deterministic algorithms. Since a classical algorithm cannot return a single value of $z$ that solves RXHOG above the given bound, it cannot return $k$ such values.

We can model a randomized algorithm as an otherwise deterministic algorithm with an initial, random bitstring, $r \in \{0,1\}^{n_r}$ for some $n_r \in \NN$ bounded by the space of the algorithm. Because $U$ is independent of $r$, it does not change the above argument.
\end{proof}